\documentclass[letter]{article}

\usepackage{pdfpages}
\usepackage{fullpage}
\usepackage{textcomp}
\usepackage{algorithm}
\usepackage{algorithmicx}
\usepackage{algpseudocode}
\usepackage{xcolor}
\usepackage{multirow}
\usepackage{amsmath}
\usepackage{amsfonts}
\usepackage{graphicx}
\usepackage{amssymb}
\usepackage{amsthm}

\newtheorem{lemma}{Lemma}
\newtheorem{lemma*}[lemma]{Lemma*}
\newtheorem{theorem}[lemma]{Theorem}
\newtheorem{theorem*}[lemma]{Theorem*}
\newtheorem{corollary}[lemma]{Corollary}

\newtheorem{observation}[lemma]{Observation}

\title{\LARGE Improved Algorithms for the Bichromatic Two-Center Problem for Pairs of Points\footnote{A preliminary version of this paper will appear in the Proceedings of the 16th Algorithms and Data Structures Symposium (WADS 2019). The work was partially done when Jie Xue was visiting Utah State University. The research of Jie Xue is partially supported by a Doctoral Dissertation Fellowship from the Graduate School of the University of Minnesota.}}

\author{
    Haitao Wang \\ Utah State University \\ \texttt{haitao.wang@usu.edu}
    \and
    Jie Xue \\ University of Minnesota \\ \texttt{xuexx193@umn.edu}
}

\date{}

\def\calI{\mathcal{I}}

\def\calU{\mathcal{U}}
\newcommand{\FVD}{\mbox{$F\!V\!D$}}

\begin{document}

\pagestyle{plain}
\pagenumbering{arabic}
\setcounter{page}{1}

\maketitle

\vspace{-0.2in}
\begin{abstract}
We consider a bichromatic two-center problem for pairs of points.
Given a set $S$ of $n$ pairs of points in the plane, for every pair, we want to assign a red color to one point and a blue color to the other, in such a way that the value $\max\{r_1,r_2\}$ is minimized, where $r_1$ (resp., $r_2$) is the radius of the smallest enclosing disk of all red (resp., blue) points.
Previously, an exact algorithm of $O(n^3\log^2 n)$ time and a $(1+\varepsilon)$-approximate algorithm of $O(n + (1/\varepsilon)^6 \log^2 (1/\varepsilon))$ time were known.
In this paper, we propose a new exact algorithm of $O(n^2\log^2 n)$ time 
and a new $(1+\varepsilon)$-approximate algorithm of $O(n + (1/\varepsilon)^3 \log^2 (1/\varepsilon))$ time. 
\end{abstract}

\section{Introduction}
\label{sec:intro}
In this paper, we consider the following bichromatic 2-center problem for pairs of points.
Given a set $S$ of $n$ pairs of points in the plane, for every pair, we want to assign a red color to one point and a blue color to the other, in such a way that the value $\max\{r_1,r_2\}$ is minimized, where $r_1$ (resp., $r_2$) is the radius of the smallest enclosing disk of all red (resp., blue) points.

Previously, Arkin et al.~\cite{ref:ArkinBi15} proposed an $O(n^3\log^2 n)$ time exact algorithm, as well as two $(1+\varepsilon)$-approximate algorithms of time $O((n/\varepsilon^2)\log n\log(1/\varepsilon))$ and $O(n+(1/\varepsilon)^6\log^2(1/\varepsilon))$, respectively.
In this paper, we propose a new exact algorithm of $O(n^2\log^2 n)$ time, which is a linear factor improvement over the exact algorithm in~\cite{ref:ArkinBi15}.
Also, we propose a new $(1+\varepsilon)$-approximate algorithm of $O(n + (1/\varepsilon)^3 \log^2 (1/\varepsilon))$ time, shaving off three $1/\varepsilon$ factors of the second term of the previous $O(n+(1/\varepsilon)^6\log^2(1/\varepsilon))$ time.


\vspace{-0.08in}
\paragraph{Related Work.}
Our problem may be considered as a new type of facility location problem. Facility location problems have been studied extensively in operations research, computational geometry, and other related areas. The classical 1-center problem for a set of points in the plane, which is also the smallest enclosing disk problem, can be solved in linear time~\cite{ref:ChazelleOn96,ref:DyerOn86,ref:MegiddoLi83}.
Our problem may be more closely related to the $2$-center problem for a set of $n$ points in the plane, which has attracted much attention. Hershberger and Suri~\cite{ref:HershbergerFi91} first solved the decision version of the problem in $O(n^2\log n)$ time, which was later improved to $O(n^2)$ time~\cite{ref:HershbergerA93}. Using this result and with parametric search technique~\cite{ref:MegiddoAp83}, Agarwal and Sharir~\cite{ref:AgarwalPl94} gave an $O(n^2\log^3 n)$ time algorithm for the planar $2$-center problem. Later, Jaromczyk and Kowaluk~\cite{ref:JaromczykAn94} proposed an $O(n^2)$ time algorithm. A breakthrough was achieved by Sharir~\cite{ref:SharirA97}, who gave the first-known subquadratic algorithm for the problem, and the running time is $O(n\log^9 n)$. Afterwards, based on Sharir's algorithm scheme~\cite{ref:SharirA97}, Eppstein~\cite{ref:EppsteinFa97} derived a randomized algorithm with $O(n\log^2 n)$ expected time, and then Chan~\cite{ref:ChanMo99} developed an $O(n\log^2 n\log^2 \log n)$ time deterministic algorithm.

As discussed in~\cite{ref:ArkinBi15}, in addition to a natural variant of the planar 2-center problem, the bichromatic 2-center problem is motivated by a chromatic clustering problem arising in certain applications in biology, e.g.,~\cite{ref:DingSo11}, as well as in transportation. For example, suppose we have a set of origin/destination pairs. We want to find two centers to build airports, such that for each origin/destination pair, we can travel from the origin to the destination by first driving to the closer airport, and then flying to the other airport, and finally driving to the destination. If the goal is to minimize the maximum of the driving time, then the problem is exactly an instance of our bichromatic 2-center problem.

The distance in our bichromatic 2-center problem is measured in the Euclidean metric.
Arkin et al.~\cite{ref:ArkinBi15} also considered the same problem in the $L_{\infty}$ metric, which is much easier and is solvable in $O(n)$ time.
In addition, instead of minimizing the maximum radius of the two smallest enclosing disks for red and blue points, Arkin et al.~\cite{ref:ArkinBi15} studied the problem of minimizing the sum of the radii of the two smallest enclosing disks. They gave an $O(n^4\log^2 n)$ time exact algorithm for this min-sum problem in the Euclidean metric, along with two $(1+\varepsilon)$-approximate algorithms, 
and an $O(n\log^2 n)$ time (deterministic) algorithm and an $O(n\log n)$ time randomized algorithm for the same problem in the $L_{\infty}$ metric.
Refer to~\cite{ref:ArkinBi15} for some other variants of the problem.


\vspace{-0.07in}
\paragraph{Outline.}
In Section~\ref{sec:pre}, we introduce some notation. 
We present our exact algorithm in Section~\ref{sec-exact} and present our approximation algorithm in Section~\ref{sec-approx}. 



\section{Preliminaries}
\label{sec:pre}

Let $r^*$ denote the radius of the larger disk in an optimal solution for our bichromatic 2-center problem.
Note that there exists an optimal solution consisting of two congruent disks of radius equal to $r^*$.
We use $OPT$ to denote such an optimal solution in which the distance between the centers of the two disks is minimized.
Let $D_1^*$ and $D_2^*$ be the two disks in $OPT$.

We say that two disks {\em bichromatically cover} $S$ if it is possible to assign a point a red color and the other a blue color for every pair of $S$ such that one disk covers all red points and the other covers all blue points.
To solve our bichromatic 2-center problem, it is sufficient to find two congruent disks of smallest radius that bichromatically cover $S$.


For a subset $S'$ of $S$, we denote by $P(S')$ the set of points in all pairs of $S'$. 

For a connected region $B$ in the plane, let $\partial B$ denote the boundary of $B$.


For any point $c$ in the plane and a value $r$, let $D_r(c)$ denote the disk centered at $c$ with radius $r$. For a set $A$ of points in the plane, define $\calI_r(A)=\bigcap_{c\in A}D_r(c)$, i.e., the common intersection of the disks $D_r(c)$ for all points $c\in A$. Note that $\calI_r(A)$ is convex and can be computed in $O(|A|\log |A|)$ time~\cite{ref:HershbergerFi91}.

For a point pair $(p,p')\in S$ and a value $r$, let $U_r(p,p')$ denote the union of the two disks $D_r(p)$ and $D_r(p')$. For a subset $S'$ of pairs of $S$, define $\calU_r(S')=\bigcap_{(p,p')\in S'}U_r(p,p')$. The following lemma, given by Arkin et al.~\cite{ref:ArkinBi15} (specifically, in Lemma 1), will be used later in our algorithm.

\begin{lemma}[{\text Arkin et al.~\cite{ref:ArkinBi15}}]
\label{lem:arkin}
Given a subset $S'$ of pairs of $S$ and a point $c$ with a value $r$ such that $D_r(c)$ covers all points of $P(S')$, $\calU_r(S')$ can be computed in $O(|S'|\log|S'|)$ time and the combinatorial complexity of $\calU_r(S')$ is $O(|S'|\cdot \alpha(|S'|))$, where $\alpha(\cdot)$ is the inverse Ackermann function.
\end{lemma}

\paragraph{Remark.}
In our algorithm, we often need to solve the following subproblem. Let $S'$, $c$, and $r$ be specified as in Lemma~\ref{lem:arkin}. Let $A$ be a set of $O(n)$ points in the plane. The problem is to determine whether  $\calU_r(S')\cap \calI_r(A)$ is empty. The problem can be solved in $O(n\log n)$ time~\cite{ref:ArkinBi15} (specifically, Lemma~1), as follows. We first compute $\calU_r(S')$ and $\calI_r(A)$ in $O(n\log n)$ time as discussed above. Then, since $\calI_r(A)$ is convex and $\calU_r(S')$ is star-shaped with respect to the point $c$, checking whether $\calU_r(S')\cap \calI_r(A)=\emptyset$ can be done in additional $O(n\alpha(n))$ time by an angular sweeping around the point $c$ (see Lemma 1 in~\cite{ref:ArkinBi15} for more details).
Note that we can also slightly change the algorithm to check whether the interior of $\calU_r(S')$ intersects the interior of $\calI_r(A)$ in the same time asymptotically as above.

\section{The Exact Algorithm} \label{sec-exact}
Before describing our algorithm in detail, we first give an overview of our approach.
To obtain the $O(n^3\log^2 n)$ time algorithm for the problem, Arkin et al.~\cite{ref:ArkinBi15} first solved in $O(n^3\log n)$ time the decision version of the problem: Given a value $r$, decide whether $r\geq r^*$. Then, an easy observation is that $r^*$ is equal to the radius of the circumcircle of two or three points of $S$, and thus one can easily form a set of $O(n^3)$ candidate values for $r^*$. Consequently, $r^*$ can be found in the set by binary search using the decision algorithm.

We take a different approach. As our problem is closely related to the planar 2-center problem for a set of points, we follow the algorithmic scheme in~\cite{ref:ChanMo99,ref:EppsteinFa97,ref:SharirA97} for the planar 2-center problem. 
More specifically, as in~\cite{ref:ChanMo99,ref:EppsteinFa97}, let $\delta^*$ be the distance of the centers of the two disks $D^*_1$ and $D^*_2$ in $OPT$.  
We consider two cases. If $\delta^*\geq r^*$, we call it the {\em distant case}; otherwise, it is the {\em nearby case}.

In the distant case, as for the planar 2-center problem~\cite{ref:ChanMo99,ref:EppsteinFa97,ref:SharirA97}, we can determine a constant number of lines such that at least one line $l$ has the following property (e.g., see Fig.~\ref{fig:casedis}):
The subset of points of $S$ on one side of $l$ (say, the left side) are contained in one disk, say, $D^*_1$, of the optimal solution, such that the subset has a point on the boundary of $D^*_1$ and $D^*_1$ is the circumcircle of two or three points of $S$. By using this observation, we first solve the decision problem of this case in $O(n^2\log n)$ time. Then, following a similar algorithm scheme to that in~\cite{ref:EppsteinFa97} and using our decision algorithm, we compute $r^*$ in $O(n^2\log^2 n)$ time  using parametric search~\cite{ref:ColeSl87,ref:MegiddoAp83}.

\begin{figure}[t]
\begin{minipage}[t]{0.49\textwidth}
\begin{center}
\includegraphics[height=1.5in]{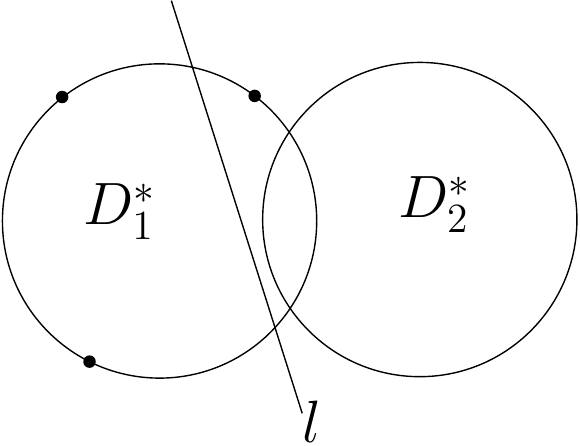}
\caption{\footnotesize Illustrating the distant case.
}
\label{fig:casedis}
\end{center}
\end{minipage}
\hspace{0.02in}
\begin{minipage}[t]{0.49\textwidth}
\begin{center}
\includegraphics[height=1.5in]{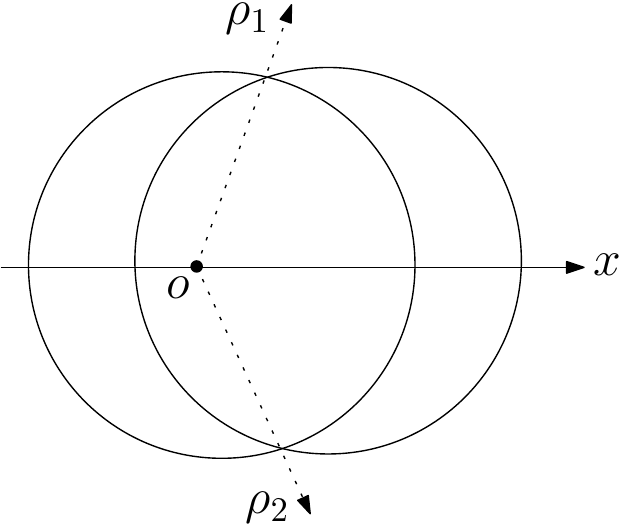}
\caption{\footnotesize Illustrating the nearby case.}
\label{fig:casenear}
\end{center}
\end{minipage}
\end{figure}

In the nearby case, as for the planar 2-center problem~\cite{ref:ChanMo99,ref:EppsteinFa97,ref:SharirA97}, we can determine a constant number of points such that at least one point $o$ is contained in the intersection of $D^*_1$ and $D^*_2$ (e.g., see Fig.~\ref{fig:casenear}). In this case, we sort all points of $S$ cyclically around $o$ and form a matrix $M$ of size $\Theta(n^2)$, such that $r^*$ is the smallest element in $M$. The similar approach is also used in~\cite{ref:ChanMo99,ref:EppsteinFa97,ref:SharirA97}. The difference, however, is that it is quite challenging to evaluate a matrix element in our problem. To this end, we first solve the decision problem in $O(n\log n)$ time and then solve the optimization problem (i.e., computing the matrix element) in $O(n\log^2 n)$ by parametric search~\cite{ref:ColeSl87,ref:MegiddoAp83}. Then, with help of an observation on the monotonicity properties of the matrix $M$, we find $r^*$ in $M$ in $O(n^2\log^2 n)$ time without evaluating all elements of $M$ (more precisely, we only need to evaluate $O(n)$ elements), by a matrix searching technique~\cite{ref:EppsteinFa97,ref:FredericksonTh82,ref:FredericksonGe84}.

Given the set $S$, because we do not know which case happens, we will simply run our algorithms for the above two cases and then return the best solution.

Comparing with the planar 2-center problem~\cite{ref:ChanMo99,ref:EppsteinFa97,ref:SharirA97}, a main challenge in our problem is that we do not have an efficient data structure to dynamically compute certain values needed in the algorithm (e.g., the elements of the matrix $M$) in poly-logarithmic time each. Instead, in most cases we have to spend more than linear time on computing each such value. This is a main obstacle that prevents us from achieving a subquadratic time algorithm for our problem.
In the next two sections, we consider the distant case and the nearby case, respectively.

\subsection{The Distant Case}
\label{sec:dist}


In this case (i.e., $\delta^*\geq r^*$), the two disks $D_1^*$ and $D^*_2$ in $OPT$ are relatively far from each other, and they may intersect or not.
As shown in~\cite{ref:EppsteinFa97}, after the smallest enclosing disk of all points of $P(S)$ is obtained, which can be done in $O(n)$ time~\cite{ref:ChazelleOn96,ref:DyerOn86,ref:MegiddoLi83}, we can determine in constant time a set $L$ of $O(1)$ lines such that at least one line $l\in L$ must have the following property: The subset $P_1$ of the points of $P(S)$ on one particular side (e.g., the left side) of $l$ are contained in one disk of $OPT$ such that a point of $P_1$ is on the boundary of the disk and the disk is the circumcircle of two or three points of $P(S)$ (e.g., see Fig.~\ref{fig:casedis}).

With $L$, because we do not know which line of $L$ and which side of the line has the above property, we will run the following algorithm for the subset $P_1$ for each side of every line of $L$, and finally return the best solution. In the following, we give our algorithm by assuming that we know the line $l$ as well as the set $P_1$ with the property stated above.

We first consider the decision problem: Given a value $r$, decide whether $r\geq r^*$. 
The property of $P_1$ leads to the following observation.

\begin{observation}\label{obser:10}
$r\geq r^*$ if and only if there exist two congruent disks of radius $r$ bichromatically covering $S$ such that one disk contains all points of $P_1$ and has one point of $P_1$ on its boundary.
\end{observation}
\begin{proof}
If there exists such a pair of congruent disks of radius $r$ as stated in the observation, then its trivially true that $r\geq r^*$.

If $r\geq r^*$, then the property of $P_1$ implies we can obtain such a pair of disks as stated in the observation by enlarging the two disks in $OPT$.
\end{proof}

We first compute the common intersection $\calI_r(P_1)$, which can be done in  $O(n\log n)$ time as discussed in Section~\ref{sec:pre}.
Then, for each point $c\in P(S)\setminus P_1$, we compute the intersection $\partial \calI_r(P_1)\cap \partial D_r(c)$, which consists of at most two points as argued in~\cite{ref:HershbergerFi91}, and can be done in $O(\log n)$ time since $\calI_r(P_1)$ is convex~\cite{ref:HershbergerFi91}. We sort these intersection points and the vertices of $\calI_r(P_1)$, along $\partial\calI_r(P_1)$, into a list $I$, which can be done in $O(n\log n)$ time since $|I|= O(n)$.

We run a scanning procedure to scan the list of $I$. For each point $c\in I$, we process it as follows.
We place a disk of radius $r$ centered at $c$, i.e., $D_r(c)$.  We wish to answer the following question: Whether do there exist two congruent disks of radius $r$ bichromatically covering $S$ such that one of them is $D_r(c)$? This can be done in $O(n\log n)$ time, as follows.

First, in $O(n)$ time, we check whether $D_r(c)$ contains at least one point from each pair of $S$. If no, then the answer to the above question is negative and the processing of the point $c$ is done (and we proceed to process the next point of $I$). Otherwise, we proceed as follows.
Let $S(c)$ be the subset of pairs of $S$ whose points are both covered by $D_r(c)$. 
Let $P(c)$ denote the subset of points of $P(S)$ not covered by $D_r(c)$.
To answer the question, it is now sufficient to determine whether there exists a disk of radius $r$ containing all points of $P(c)$ and at least one point from each pair of $S(c)$.
To this end, we first compute $\calU_r(S(c))$, which can be done in $O(n\log n)$ time by Lemma~\ref{lem:arkin},
since every point of $S(c)$ is covered by $D_r(c)$. Next, we compute $\calI_r(P(c))$ in $O(n\log n)$ time. Finally, we determine whether $\calU_r(S(c))\cap \calI_r(P(c))$ is empty, which can be done in $O(n\log n)$ time as remarked in Section~\ref{sec:pre}. Note that the answer to our question is positive if and only if $\calU_r(S(c))\cap \calI_r(P(c))$ is not empty.

If the answer to our question is positive, then we stop our decision algorithm with the assertion that $r\geq r^*$, in which case two congruent disks of radius $r$ that bichromatically cover $S$ are also obtained as implied by the above algorithm. Otherwise, we continue on the next point of $I$. If the answer to the question is negative for all points of $I$, then we stop with the assertion that $r< r^*$.
Observation~\ref{obser:10} guarantees the correctness of the algorithm.

Since $|I|=O(n)$ and processing each point of $I$ takes $O(n\log n)$ time, the total time of the  algorithm is $O(n^2\log n)$.


With the decision algorithm, in Lemma~\ref{lem:distantopt} we solve the optimization problem, i.e., computing $r^*$, in $O(n^2\log^2 n)$ time using parametric search~\cite{ref:ColeSl87,ref:MegiddoAp83}. The parametric search scheme is almost the same as that in~\cite{ref:EppsteinFa97} (i.e., in Section 4). 

\begin{lemma}\label{lem:distantopt}
    An optimal solution can be computed in $O(n^2\log^2 n)$ time. 
\end{lemma}
\begin{proof}
We use a parametric search scheme that is almost the same as that in~\cite{ref:EppsteinFa97} (i.e., in Section 4). For completeness, we briefly discuss it below.

First of all, although we do not know $r^*$, we need to determine the combinatorial structure of $\calI_{r^*}(P_1)$, i.e., the points of $P_1$ that define the arcs of $\partial\calI_{r^*}(P_1)$.
For this, notice that if we place a disk $D$ of radius $r^*$ centered at a vertex of $\calI_{r^*}(P_1)$, then $D$ has two points of $P_1$ on its boundary and contains all other points of $P_1$. Therefore, such a disk $D$ has its center on an edge of the farthest Voronoin diagram of the points of $P_1$, denoted by $\FVD(P_1)$. As such, we compute $\FVD(P_1)$, and for each edge $e$ of $\FVD(P_1)$, we determine an interval (for the parameter $r$ in the decision algorithm), in the following way. Note that $e$ is actually a half-line~\cite{ref:deBergCo08}, which is on a bisector of two points $p$ and $q$ of $P_1$. Suppose $p_e$ is the only endpoint of $e$. We keep the following interval $[r_e,\infty)$ for $e$, where $r_e$ is equal to the distance between $p_e$ and $p$. Observe that if we place a disk $D$ of radius $r\in [r_e,\infty)$ at a point of $e$ whose distance from $p$ is $r$, then $D$ will cover all points of $P_1$ and have both $p$ and $q$ on its boundary.

As $\FVD(P_1)$ has $O(n)$ edges, we can determine $O(n)$ intervals $[r_e,\infty)$ as above. Then, we sort all those $r_e$ values, and do binary search on the sorted list using our decision algorithm to obtain an interval $(r_1,r_2]$ that contains $r^*$. This steps takes $O(n^2\log^2 n)$ time because the decision algorithm is called $O(\log n)$ times.
Next, we pick any value $r\in (r_1,r_2]$ and construct $\calI_r(P_1)$, which has the same combinatorial structure as $\calI_{r^*}(P_1)$ according to the above discussion.

The decision algorithm would perform a scanning procedure on $\partial\calI_{r^*}(P_1)$, if we knew $r^*$. For this, we need to determine the combinatorial intersections of $\partial \calI_{r^*}(P_1)$ and $\partial D_{r^*}(p)$ for every $p\in P(S)\setminus P_1$ by binary search. More specifically, we need to determine the at most two points of $P_1$ such that the intersections $\partial \calI_{r^*}(P_1)$ and $\partial D_{r^*}(p)$ lie on the arcs of $\partial \calI_{r^*}(P_1)$ defined by the two points.
All these intersections can be computed by $O(n)$ parallel binary search operations.
Afterwards, we sort these intersections, along with the vertices of $\calI_{r^*}(P_1)$.
All these are suitable for Cole's parametric search speedup~\cite{ref:ColeSl87},
which takes $O(n\log n)$ time, in addition to $O(\log n)$ calls on the decision algorithm.
It is shown in \cite{ref:EppsteinFa97} (Lemma 4.2) the behavior of the parametric search undergoes a discrete change at $r=r^*$, meaning that $r=r^*$ will be tested by the decision algorithm. Therefore, among all values $r$ larger than or equal to $r^*$ tested by the decision algorithm in the parametric search, the smallest one is $r^*$.

As such, the total time for computing $r^*$ is $O(n^2\log^2 n)$. After $r^*$ is computed, we can apply our decision algorithm with $r=r^*$ to find two congruent disks of radius $r$ that bichromatically cover $S$ as our optimal solution.
\end{proof}

\subsection{The Nearby Case}
\label{sec:near}


In this case (i.e., $\delta^*<r^*$), the centers of the two disks $D_1^*$ and $D_2^*$ of $OPT$ are relatively close and the two disks must intersect. As shown in~\cite{ref:EppsteinFa97,ref:SharirA97}, after the smallest enclosing disk of $P(S)$ is computed, we can determine in constant time a set of $O(1)$ points such that one point $o$ must be in $D_1^*\cap D^*_2$. Because we do not know which point has the property, we will run the following algorithm for each such point as $o$, and then return the best solution. In the following, we assume that the point $o$ has the property. We make $o$ as the origin of the plane.

Note that $\partial D^*_1$ and $\partial D^*_2$ have exactly two intersections, and let $\rho_1$ and $\rho_2$ be the two rays through these intersections emanating from $o$ (e.g., see Fig.~\ref{fig:casenear}). As argued in~\cite{ref:ChanMo99}, one of the two coordinate axes must separate $\rho_1$ and $\rho_2$ since the angle between the two rays lies in $[\pi/2,3\pi/2]$, and without loss of generality, we assume it is the $x$-axis. Again, because we do not know which axis separates the two rays, we will run the following algorithm once for the $x$-axis and once for the $y$-axis, and then return the best solution. In the following, we present the algorithm by assuming that it is the $x$-axis.

For ease of exposition, we make a general position assumption that no point of $P(S)$ has the same $y$-coordinate as $o$ and no two points of $P(S)$ are collinear with $o$. The degenerate case can still be solved by our technique, but the discussion would be more tedious.

Let $P^+$ denote the subset of points of $P(S)$ above the $x$-axis, and $P^-$ the subset below the $x$-axis. To simplify the discussion, let $|P^+|=|P^-|=n$. Let $p_1,p_2,\ldots, p_n$ be the sorted list of the points of $P^+$ counterclockwise around $o$, and $q_1,q_2,\ldots, q_n$ the sorted list of the points of $P^-$ also counterclockwise around $o$ (e.g., see Fig.~\ref{fig:sort}).
For each $i=0,1,\ldots, n$ and $j=0,1,\ldots,n$, define $L_{ij}=\{p_{i+1}\ldots, p_n, q_1,\ldots,q_j\}$ and $R_{ij}=\{q_{j+1},\ldots,q_n,p_1,\ldots,p_i\}$. Note that if $i=n$, then $L_{ij}=\{q_1,\ldots,q_j\}$, and if $j=n$, then $R_{ij}=\{p_1,\ldots,p_i\}$. In other words, if we consider a ray emanating from $o$ and between $p_i$ and $p_{i+1}$ and another ray emanating from $o$ and between $q_j$ and $q_{j+1}$, then $L_{ij}$ (resp., $R_{ij}$) consisting of all points to the left (resp., right) of the two rays (e.g., see Fig.~\ref{fig:sort}).

\begin{figure}[t]
\begin{minipage}[t]{\textwidth}
\begin{center}
\includegraphics[height=1.7in]{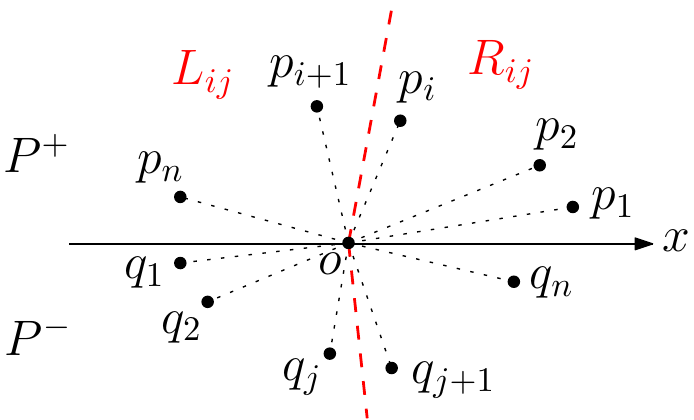}
\caption{\footnotesize Illustrating the points of $P^+$ and $P^-$.}
\label{fig:sort}
\end{center}
\end{minipage}
\end{figure}

For any pair $(i,j)$ with $0\leq i, j\leq n$, we consider the following {\em restricted bichromatic 2-center problem}. Find a pair of two congruent disks $D_1$ and $D_2$ of the smallest radius such that the following hold: (1) $D_1$ and $D_2$ bichromatically cover $S$; (2) $D_1$ covers all points of $L_{ij}\cup \{o\}$ and $D_2$ covers all points of $R_{ij}\cup \{o\}$. We let $r^*_{ij}$ denote the radius of the two disks in an optimal solution. We use $RB2C(i,j)$ to refer to the problem. If a pair of disks satisfies the above two conditions, then we call them a {\em feasible} pair of disks for $RB2C(i,j)$.

The following lemma shows why we need to consider the problem $RB2C(i,j)$.

\begin{lemma}\label{lem:opt}
$r^*=\min_{0\leq i, j\leq n} r^*_{ij}$.
\end{lemma}
\begin{proof}
First of all, $r^*\leq r^*_{ij}$ holds for any $(i,j)$ with $0\leq i, j\leq n$, for $r^*_{ij}$ is the radius of two congruent disks that bichromatically cover $S$. Hence,
$r^*\leq \min_{0\leq i, j\leq n} r^*_{ij}$. Below, we prove that
$r^*\geq \min_{0\leq i, j\leq n} r^*_{ij}$. It is
sufficient to show that $r^*\geq r^*_{ij}$ for some pair $(i,j)$.

Recall that the two optimal disks $D_1^*$ and $D_2^*$ of $OPT$ both contain the
point $o$.
Suppose the ray $\rho_1$ is between the two points $p_{i+1}$ and $p_i$
in $P^+$ (let $i=0$ if all points of $P^+$ are to the left of
$\rho_1$, and $i=n$ if all points of $P^+$ are to the right of
$\rho_1$); similarly, suppose $\rho_2$ is between two points $q_{j}$
and $q_{j+1}$. According to the definitions of $\rho_1$ and $\rho_2$,
one disk of $OPT$, say $D^*_1$, contains all points of $L_{ij}$, and the
other disk $D^*_2$ contains all points of $R_{ij}$. Further, $o\in D_1^*\cup D_2^*$. Therefore, the two
disks of $OPT$ form a feasible pair for the problem $RB2C(i,j)$. Because $r^*$ is the radius of $D^*_1$ and $D^*_2$, $r^*\geq r^*_{ij}$ must hold.
\end{proof}

Define an $(n+1)\times (n+1)$ matrix $M[0\ldots n; 0\ldots n]$, where
$M[i,j]=r^*_{ij}$ for all $0\leq i, j\leq n$.
By Lemma~\ref{lem:opt}, $r^*$ is equal to the minimum element in $M$.
To find $r^*$ from $M$, instead of computing all $(n+1)^2$ elements of $M$, we will prove certain monotonicity properties of the
matrix and then apply a matrix searching
technique~\cite{ref:EppsteinFa97,ref:FredericksonTh82,ref:FredericksonGe84}, so that it suffices to compute $O(n)$ elements of $M$.
One of the challenges is that it is not trivial to compute even a single element
of $M$. In the following, we first present an algorithm 
that can compute a
single matrix element $M[i,j]$, i.e., $r^*_{ij}$, in $O(n\log^2 n)$ time. By using the algorithm, we describe later 
how to find $r^*$ in $M$ in $O(n^2\log^2 n)$ time.

\subsubsection{An Algorithm for Computing $r^*_{ij}$}
\label{sec:comrij}

To compute $r^*_{ij}$ for $RB2C(i,j)$, we will resort to parametric search again. To this end, we first solve the decision problem: Given a value $r$, decide whether $r\geq r^*_{ij}$. We present an $O(n\log n)$ time decision algorithm for it.
Let $S_1$ be the subset of pairs of $S$ whose points are both in $L_{ij}$, and $S_2$ the subset of pairs of $S$ whose points are both in $R_{ij}$. The following observation is self-evident.

\begin{observation}\label{obser:20}
$r\geq r_{ij}^*$ if and only if there exist a pair of congruent disks of radius $r$ such that one disk covers all points of $L_{ij}\cup \{o\}$ and at least one point from each pair of $S_2$, and the other disk covers all points of $R_{ij}\cup \{o\}$ and at least one point from each pair of $S_1$.
\end{observation}

Based on Observation~\ref{obser:20}, our algorithm works as follows. First, we compute the radius $r_1$ of the smallest enclosing disk of $L_{ij}\cup \{o\}$ and the radius $r_2$ of the smallest enclosing disk of $R_{ij}\cup \{o\}$. Note that Observation~\ref{obser:20} implies that $r_1\leq r^*_{ij}$ and $r_2\leq r^*_{ij}$. Hence, if $r<\max\{r_1,r_2\}$, then we have $r<r^*_{ij}$, and thus we can stop the algorithm. Otherwise, we proceed as follows.

Observe that there exists a disk of radius $r$ covering all points of $L_{ij}\cup \{o\}$ and at least one point from each pair of $S_2$ if and only if $\calI_r(L_{ij}\cup \{o\})\cap \calU_r(S_2)\neq \emptyset$.
Computing $\calI_r(L_{ij}\cup \{o\})$ can be done in $O(n\log n)$ time. For $\calU_r(S_2)$, notice that every point of $S_2$ is in the disk $D_r(c)$ because $r_2\leq r$, where $c$ is the center of the smallest enclosing disk of $R_{ij}\cup \{o\}$. Hence, by Lemma~\ref{lem:arkin}, $\calU_r(S_2)$ can be computed in $O(n\log n)$ time. In addition, determining whether $\calI_r(L_{ij}\cup \{o\})\cap \calU_r(S_2)= \emptyset$ can also be done in $O(n\log n)$ time, as remarked in Section~\ref{sec:pre}. As such, determining whether there exists a disk of radius $r$ covering all points of $L_{ij}\cup \{o\}$ and at least one point from each pair of $S_2$ can be done in $O(n\log n)$ time.

Similarly, it takes $O(n\log n)$ time to determine whether there exists a disk of radius $r$ covering all points of $R_{ij}\cup \{o\}$ and at least one point from each pair of $S_1$.
This solves the decision problem in $O(n\log n)$ time.

With the above decision algorithm in hand, we can apply parametric search to compute $r_{ij}^*$.
We first prove the following lemma, which is critical to our parametric search algorithm. 

\begin{lemma}\label{lem:30}
Suppose $D_1$ and $D_2$ are a pair of congruent disks in an optimal solution for the problem $RB2C(i,j)$ with $D_1$ containing all points of $L_{ij}\cup \{o\}$ and $D_2$ containing all points of $R_{ij}\cup \{o\}$. Then, either $L_{ij}$ has a point on $\partial D_1$ or $R_{ij}$ has a point on $\partial D_2$.
\end{lemma}
\begin{proof}
If $L_{ij}$ has a point on $\partial D_1$, then the lemma obviously follows. In the following, we assume that $L_{ij}$ does not have a point on $\partial D_1$.

Let $P_1$ be the subset of points of $P(S)$ contained in $D_1$, and let $r_1$ be the radius of the smallest enclosing disk of $P_1$. Let $P_2$ be the subset of points of $P(S)$ contained in $D_2$, and let $r_2$ be the radius of the smallest enclosing disk of $P_2$. Note that $r_{ij}^*$ is the radius of $D_1$ and $D_2$. Hence, $r_1\leq r_{ij}^*$ and $r_2\leq r_{ij}^*$.

Depending on whether $r_1<r_{ij}^*$, there are two cases.

\begin{enumerate}
  \item If $r_1<r_{ij}^*$, then it must hold that $r_2=r_{ij}^*$ since otherwise we could find a feasible pair of congruent disks for $RB2C(i,j)$ with radius smaller than $r_{ij}^*$, contradicting with the definition of $r^*_{ij}$. Hence, $r_2>r_1$.

We claim that $R_{ij}$ must have a point on $\partial D_2$. Assume to the contrary that this is not true.
Then, all points of $A_2$ must be in $L_{ij}\cup\{o\}$, where $A_2$ is the subset of points of $P_2$ on the boundary of $D_2$. Since $r_2=r_{ij}^*$ and $r_2$ is the radius of the smallest enclosing disk of $P_2$, $D_2$ is the smallest enclosing disk of $P_2$. By the definition of $A_2$, $D_2$ is the smallest enclosing disk of $A_2$. Since $A_2\subseteq L_{ij}\cup\{o\}\subseteq P_1$ and $r_1$ is the radius of the smallest enclosing disk of $P_1$, we obtain that $r_1\geq r_2$, which contradicts with the fact that $r_2>r_1$.

This proves that $R_{ij}$ must have a point on $\partial D_2$.

\item
If $r_1=r_{ij}^*$, then since $r_1$ is the radius of the smallest
enclosing disk of $P_1$, $D_1$ is the smallest enclosing disk of
$P_1$. Let $A_1$ be the subset of points of
$P_1$ on the boundary of $D_1$. Note that $|A_1|\geq 2$.
Since $L_{ij}$ does not have a point
on $\partial D_1$, $A_1\subseteq R_{ij}\cup\{o\}$. Thus, $A_1\subseteq
P_2$, for $R_{ij}\cup\{o\}\subseteq P_2$.

Note that $D_1$ is the smallest enclosing disk of $A_1$, with radius
$r_1=r_{ij}^*$ and all points of $A_1$ on the boundary of $D_1$.
Because the radius of $D_2$ is $r_{ij}^*$, $D_2$
covers all points of $P_2$, and $A_1\subseteq P_2$, the points of
$A_1$ must be all on the boundary of $D_2$. As $A_1\geq 2$ and
$A_1\subseteq R_{ij}\cup\{o\}$, we obtain that at least one point of $R_{ij}$ must be
on the boundary of $D_2$.
\end{enumerate}

Therefore, in either case above, $R_{ij}$ must have a point on $\partial D_2$.
The lemma thus follows.
\end{proof}

In light of Lemma~\ref{lem:30}, without loss of generality, we assume that $L_{ij}$ has a point on $\partial D_1$ (because we do not know which case actually happens, we will run the following algorithm twice: once for $L_{ij}$ and once for $R_{ij}$, and then return the best solution). Based on this property, we show below that we can compute $r_{ij}^*$ using the same parametric search scheme as in the algorithm for Lemma~\ref{lem:distantopt}.
To this end, 
we need to make sure a corresponding decision algorithm, i.e., the one in the algorithm for Lemma~\ref{lem:distantopt}, still works for our problem $RB2C(i,j)$. In the following, we present such a decision algorithm, i.e., for determining whether $r\geq r^*_{ij}$ for any given value $r$.

Let $P_1=L_{ij}\cup\{o\}$.
We compute $\calI_r(P_1)$. Then, for each point $c\in R_{ij}$, we compute the at most two intersection points of $\partial \calI_r(P_1)\cap \partial D_r(c)$. We sort these points and the vertices of $\calI_r(P_1)$ into a list $I$.

We scan the list $I$. For each point $c\in I$, we process it as follows.
We place a disk $D_r(c)$ of radius $r$ centered at $c$. Note that $D_r(c)$ covers all points of $P_1$.  We wish to determine whether there exists a feasible pair of disks for $RB2C(i,j)$ such that one disk is $D_r(c)$. To this end, we first check whether $D_r(c)$ contains at least one point from each pair of $S$ (note that $D_r(c)$ already contains $o$). If no, the processing of the point $c$ is done. Otherwise, we proceed as follows. Let $S(c)$ be the subset of pairs of $S$ whose points are both covered by $D_r(c)$. Let $P(c)$ be the subset of points of $P(S)$ not covered by $D_r(c)$, and we also add the point $o$ to $P(c)$.
It suffices to determine whether there exists a disk of radius $r$ containing all points of $P(c)$ and at least one point from each pair of $S(c)$. For this, we can compute $\calI_r(P(c))$ and $\calU_r(S(c))$, and then determine whether their intersection is empty.

The above algorithm solves the decision problem, although it takes more time than our previous $O(n\log n)$ time solution. The algorithm implies that we can use the same parametric search scheme as the one for Lemma~\ref{lem:distantopt}.
Specifically, we first compute the farthest Voronoi diagram $\FVD(P_1)$ of $P_1$, and then use it to determine the combinatorial structure of $\calI_{r_{ij}^*}(P_1)$. Next, we compute the sorted list $I$ with respect to $r_{ij}^*$ by using Cole's parametric search.
As the algorithm for Lemma~\ref{lem:distantopt}, the parametric search undergoes a discrete change at $r=r_{ij}^*$ (again, see Lemma~4.2 in \cite{ref:EppsteinFa97} for a similar proof), meaning that $r=r_{ij}^*$ will be tested by the decision algorithm. Therefore,
among all values $r$ larger than or equal to $r_{ij}^*$ tested by the decision algorithm in the above parametric search, the smallest one is $r_{ij}^*$. The overall algorithm takes $O(n\log n)$ time, in addition to $O(\log n)$ calls on the decision algorithm. If we use our previous $O(n\log n)$ time decision algorithm, the total time for computing $r_{ij}^*$ is $O(n\log^2 n)$.
\begin{lemma}\label{lem:40}
For any pair $(i,j)$ with $0\leq i, j\leq n$, the value $r_{ij}^*$ can be computed in $O(n\log^2 n)$ time.
\end{lemma}

After $r_{ij}^*$ is computed, one can apply our decision algorithm with $r=r^*_{ij}$ to find a pair of optimal congruent disks of radius equal to $r^*_{ij}$, in additional $O(n\log n)$ time.

\subsubsection{Searching $r^*$ in the Matrix $M$}
\label{sec:compr}

We now find $r^*$ in $M$ by using the algorithm in the previous subsection. The runtime of our algorithm is $O(n^2\log^2 n)$.

To find $r^*$ in $M$, a straightforward way is to compute all $(n+1)^2$ elements of $M$ and then return the minimum one, which would take $O(n^3\log^2 n)$ time by Lemma~\ref{lem:40}. To reduce the time, we resort to some matrix searching techniques~\cite{ref:EppsteinFa97,ref:FredericksonTh82,ref:FredericksonGe84}. To this end, we need some sort of ``stronger'' solution for the problem $RB2C(i,j)$, as follows.

Consider any pair $(i,j)$ with $0\leq i, j\leq n$.
Define $S_1$ to be the subset of pairs of $S$ whose points are both in $L_{ij}$. 
Similarly, define $S_2$ to be the subset of pairs of $S$ whose points are both in $R_{ij}$. Define $D^1_{ij}$ to be the smallest disk containing all points of $L_{ij}\cup\{o\}$ and at least one point of each pair of $S_2$, and let $l_{ij}$ be the radius of $D^1_{ij}$. Similarly, define $D^2_{ij}$ to be the smallest disk containing all points of $R_{ij}\cup\{o\}$ and at least one point of each pair of $S_1$, and let $r_{ij}$ be the radius of $D^2_{ij}$.
We have the following observation.

\begin{observation}\label{obser:30}
\begin{enumerate}
\item
$\max\{l_{ij},r_{ij}\}= r_{ij}^*$.
\item
If $l_{ij}<r^*_{ij}$, then $r_{ij}=r^*_{ij}$; otherwise, $l_{ij}=r^*_{ij}$.
\end{enumerate}
\end{observation}
\begin{proof}
Notice that $D^1_{ij}$ and $D^2_{ij}$ form a feasible pair of disks for the problem $RB2C(i,j)$. Therefore, $\max\{l_{ij},r_{ij}\}\geq r_{ij}^*$ holds.

Next, we show that $\max\{l_{ij},r_{ij}\}\leq r_{ij}^*$.
Consider an optimal solution for the problem $RB2C(i,j)$, in which one disk $D_1$ must contain all points of $L_{ij}\cup\{o\}$ and at least one point of each pair of $S_2$ and the other disk $D_2$ must contain all points of $R_{ij}\cup\{o\}$ and at least one point of each pair of $S_1$, and $D_1$ and $D_2$ are congruent with radius $r_{ij}^*$. By the definitions of $l_{ij}$ and $r_{ij}$, $l_{ij}\leq r_{ij}^*$ and $r_{ij}\leq r_{ij}^*$. Therefore, $\max\{l_{ij},r_{ij}\}\leq r_{ij}^*$.

The above proves that $\max\{l_{ij},r_{ij}\}= r_{ij}^*$, from which the second part of the observation easily follows.
\end{proof}

Our algorithm for searching $r^*$ in $M$ needs to solve the following subproblem: decide whether $l_{ij}<r^*_{ij}$. With help of Lemma~\ref{lem:40}, we have the following result.

\begin{lemma}\label{lem:50}
For any $(i,j)$ with $0\leq i, j\leq n$, deciding whether $l_{ij}<r^*_{ij}$ can be done in $O(n\log^2 n)$ time.
\end{lemma}
\begin{proof}
We first compute $r_{ij}^*$ by Lemma~\ref{lem:40}.
Let $r=r_{ij}^*$, and $P_1=L_{ij}\cup \{o\}$.
We compute $\calI_r(P_1)$ and $\calU_r(S_2)$. $\calI_r(P_1)$ can be computed in $O(n\log n)$ time,  as discussed in Section~\ref{sec:pre}. For $\calU_r(S_2)$, notice that
all points of $S_2$ are covered by the disk $D_r(c)$, where $c$ is the center of the smallest enclosing disk of $R_{ij}$, since $r=r_{ij}^*$ is no smaller than the radius of the smallest enclosing disk of $R_{ij}$.
Hence, once $c$ is computed in $O(n)$ time~\cite{ref:ChazelleOn96,ref:DyerOn86,ref:MegiddoLi83}, $\calU_r(S_2)$ can be computed in $O(n\log n)$ time by Lemma~\ref{lem:arkin}.
Then, observe that $l_{ij}<r$ if and only if the intersection of the interior of $\calI_r(P_1)$ and the interior of $\calU_r(S_2)$ is not empty. Checking whether the interior of $\calI_r(P_1)$ intersects the interior of $\calU_r(S_2)$ can be done in additional $O(n\alpha(n))$ time as remarked in Section~\ref{sec:pre}. Hence, we can determine whether $l_{ij}<r^*_{ij}$ in $O(n\log^2 n)$ time, which is dominated by the algorithm for computing $r_{ij}^*$.
\end{proof}

The following lemma provides a basis for applying a matrix searching technique~\cite{ref:EppsteinFa97,ref:FredericksonTh82,ref:FredericksonGe84} to search $r^*$ in the matrix $M$.

\begin{lemma}\label{lem:60}
For any $0\leq i,j\leq n$, if $l_{ij}<r^*_{ij}$, then $r^*_{ij}\leq r^*_{i'j'}$ for any $i'\in [i,n]$ and $j'\in [0,j]$; otherwise, $r^*_{ij}\leq r^*_{i'j'}$ for any $i'\in [0,i]$ and $j'\in [j,n]$.
\end{lemma}
\begin{proof}
If $l_{ij}<r^*_{ij}$, then $r_{ij}=r_{ij}^*$ by Observation~\ref{obser:30}. Consider any pair $(i',j')$ with $i'\in [i,n]$ and $j'\in [0,j]$. By their definitions, $R_{ij}\subseteq R_{i'j'}$ and $L_{i'j'}\subseteq L_{ij}$.
Let $S_1'$ be the subset of pairs of $S$ whose points are both in $L'_{ij}$.
Then, $S_1'\subseteq S_1$, for  $L_{i'j'}\subseteq L_{ij}$.
We claim that $r_{ij}\leq r_{i'j'}$. Indeed, consider a disk $D$ of radius $r_{i'j'}$ containing all points of $R_{i'j'}$ and at least one point for each pair of $S'_1$. Observe that at least one point of each pair of $S_1\setminus S_1'$ is in $R_{i'j'}$. Because $R_{ij}\subseteq R_{i'j'}$, the disk $D$ contains all points of $R_{ij}$ and at least one point of $S_1$. Therefore, by the definition of $r_{ij}$, since $r_{i'j'}$ is the radius of $D$, $r_{ij}\leq r_{i'j'}$ holds.
Because $r_{ij}=r_{ij}^*$ and $r_{i'j'}\leq r^*_{i'j'}$ (by Observation~\ref{obser:30}), we obtain that $r^*_{ij}\leq r^*_{i'j'}$.

If $l_{ij}<r^*_{ij}$ does not hold, then $l_{ij}=r^*_{ij}$ by Observation~\ref{obser:30}. This is a symmetric case to the above and by a similar proof we can show that $r^*_{ij}\leq r^*_{i'j'}$ holds for any $i'\in [0,i]$ and $j'\in [j,n]$.
\end{proof}

Recall that $r^*$ is equal to the smallest element of $M$ and each matrix element $M[i,j]$ is equal to $r_{ij}^*$. Lemma~\ref{lem:60} essentially tells the following: If $l_{ij}<r^*_{ij}$ for a cell $M[i,j]$ of $M$, then all cells of $M$ to the southwest of $M[i,j]$ can be pruned (i.e., they are irrelevant to finding $r^*$); otherwise all cells of $M$ to the northeast of $M[i,j]$ can be pruned. This is exactly the property the matrix searching algorithm in~\cite{ref:EppsteinFa97} (i.e., the algorithm in Lemma~5.3, which relies on the property in Lemma~5.2 that is similar to ours and follows a similar technique as in~\cite{ref:FredericksonTh82,ref:FredericksonGe84}) relies on.
By using that algorithm, we can compute $r^*$ from $M$ with $O(n)$ matrix cell evaluations and $O(n)$ additional time, and here each matrix cell evaluation on $M[i,j]$ is to compute $r_{ij}^*$ and determine whether $l_{ij}<r^*_{ij}$. By Lemmas~\ref{lem:40} and \ref{lem:50}, each matrix cell evaluation can be done in $O(n\log^2 n)$ time, which leads to an $O(n^2\log^2 n)$ time algorithm for finding $r^*$ from the matrix $M$.

Once $r^*$ is known, we can obtain a pair of optimal disks as follows. Assume that $r^*$ is equal to $r^*_{ij}$ for some $i$ and $j$. We apply our decision algorithm with $r=r^*$ for the problem  $RB2C(i,j)$ to obtain two congruent disks of radius $r^*$ as the optimal solution for our original bichromatic 2-center problem on $S$.

\begin{theorem}
The bichromatic 2-center problem on a set of $n$ pairs of points in the plane is solvable in $O(n^2\log^2 n)$ time.
\end{theorem}

\section{The Approximation Algorithm} \label{sec-approx}
In this section, we give a $(1+\varepsilon)$-approximate algorithm of  $O(n+(1/\varepsilon)^3\log^2(1/\varepsilon))$ time, improving the $O(n+(1/\varepsilon)^6\log^2(1/\varepsilon))$-time algorithm of \cite{ref:ArkinBi15}.
We assume that $\varepsilon$ is sufficiently small.

\subsection{Reducing to the IB2C problem}
The first step of our algorithm is to use a grid and identify the points in the same cell.
This is similar to an idea in \cite{ref:ArkinBi15} (and is in fact a standard technique used in many other geometric approximation algorithms), and here we describe it in a self-contained way.
Let $\tilde{r}$ be the radius of the minimum enclosing disk of $P(S)$.
Clearly, $\tilde{r} \geq r^*$.
If $\tilde{r} \geq 10r^*$, the problem is actually easy.
\begin{lemma}\label{lem:70}
    If $\tilde{r} \geq 10r^*$, then the bichromatic 2-center problem for $S$ can be solved exactly in $O(n)$ time.
\end{lemma}
\begin{proof}
Consider the two disks $D^*_1$ and $D^*_2$ in an optimal solution $OPT$.
Since $r^* \leq \tilde{r}/10$, the distance between the centers of $D^*_1$ and $D^*_2$ is at least $8r^*$ (otherwise $P(S)$ would be contained in a disk of radius less than $10r^*$, contradicting the fact $\tilde{r} \geq 10r^*$).
In this case, it was shown in~\cite{ref:EppsteinFa97full} (the full version of \cite{ref:EppsteinFa97}) that we can find in $O(n)$ time a constant number of lines such that one line separates $D^*_1$ and $D^*_2$. Then, for each such line, we can compute the smallest enclosing disk of the points in each side of the line and return the larger radius as the solution for the line. Finally, the best solution of these lines is an optimal solution for our problem. The total time of the algorithm is $O(n)$. 
\end{proof}

So it suffices to consider the case where $\tilde{r} \in [r^*,10r^*)$.
We build a grid $G$ consisting of square cells of side-length $\delta=\varepsilon\tilde{r}/100$.
For a point $x \in \mathbb{R}^2$, we denote by $\Box_x$ the cell containing $x$.
For each pair $(a_i,a_i') \in S$, we create another point-pair $(b_i,b_i')$ where $b_i$ is an arbitrary vertex of $\Box_{a_i}$ and $b_i'$ is an arbitrary vertex of $\Box_{a_i'}$.
Let $S'$ be the set of all these pairs excluding the duplicates, and $\mathcal{D}$ be the collection of all disks whose centers are grid points of $G$.
Consider the bichromatic 2-center problem for $S'$ with the solution space $\mathcal{D}$ (i.e., the disks must be chosen from $\mathcal{D}$), and let $(D_r(c_1),D_r(c_2))$ be an optimal solution of this problem consisting of two congruent disks of radius $r$.
Set $r' = (1+\varepsilon/3)r$.
\begin{lemma}\label{lem:80}
	$(D_{r'}(c_1),D_{r'}(c_2))$ is a feasible solution for the bichromatic 2-center problem for $S$.
	Furthermore, $r^* \leq r' \leq (1+\varepsilon) r^*$.
\end{lemma}
\begin{proof}
We first notice that $2\delta = \varepsilon\tilde{r}/50 \leq \varepsilon r^*/5$.
If we compare each pair $(a_i,a_i') \in S$ with its corresponding pair $(b_i,b_i') \in S'$, we have $\lVert a_i - b_i \rVert_2 \leq 2\delta \leq \varepsilon r^*/5$ and $\lVert a_i' - b_i' \rVert_2 \leq 2\delta \leq \varepsilon r^*/5$.
Therefore, if we have a solution for the bichromatic 2-center problem for $S$, expanding both disks in the solution with an additive factor $\varepsilon r^*/5$ gives us a solution for the bichromatic 2-center problem for $S'$, and vice versa.
This implies $r^*-\varepsilon r^*/5 \leq r \leq r^*+\varepsilon r^*/5$.
Since $\varepsilon$ is sufficiently small, we have $3 r^*/5 \leq r$.
Then $r' = r+\varepsilon r/3 \geq r+ \varepsilon r^*/5$.
By the argument above, $(D_{r'}(c_1),D_{r'}(c_2))$ is a feasible solution for the bichromatic 2-center problem for $S$.
As such, $r' \geq r^*$.
On the other hand, we have $r' = (1+\varepsilon/3)r \leq (1+\varepsilon/3)(1+\varepsilon/5)r^* \leq (1+\varepsilon) r^*$.
\end{proof}

The above lemma reduces the approximate bichromatic 2-center problem for $S$ to the (exact) bichromatic 2-center problem for $S'$ with the solution space $\mathcal{D}$.
To solve the latter problem, we exploit its following special properties.
\begin{itemize}
    \item All points in $P(S')$ are grid points of $G$.
    \item All points in $P(S')$ are contained in a (orthogonal) square of side-length $(2+2\varepsilon)\tilde{r}$.
    Indeed, the diameter of $P(S)$ is at most $2\tilde{r}$ and thus the diameter of $P(S')$ is at most $(2+2\varepsilon)\tilde{r}$.
    \item The two disk-centers in a solution must be grid points of $G$.
\end{itemize}
By scaling, we may assume that the grid points of $G$ are the points in $\mathbb{R}^2$ with integral coordinates (or \textit{integral points} hereafter).
We say a disk is \textit{integral} if its center is an integral point.
We then pass to the following \textit{integral} bichromatic 2-center (IB2C) problem.
\medskip

\noindent
\textbf{The integral bichromatic 2-center problem (IB2C).}
Given a set $T$ of $m$ point-pairs each consisting of two integral points in $[U] \times [U]$ where $[U] = \{1,\dots,U\}$, find two integral disks bichromatically covering $T$ such that the radius of the larger one is minimized.
\medskip

\noindent
We have $|S'| \leq n$.
Before scaling, the points in $P(S')$ are contained in a square of side-length $(2+2\varepsilon)\tilde{r}$ and the side-length of the cells in $G$ is $\Theta(\varepsilon\tilde{r})$.
Therefore, the original problem is reduced to the IB2C problem with $m = O(n)$ and $U = O(1/\varepsilon)$.
$S'$ can be computed in $O(n)$ time from $S$ (using the floor function, as did in \cite{ref:ArkinBi15}).
Hence, if the IB2C problem can be solved in $f(m,U)$ time, then there is an $O(n+f(n,1/\varepsilon))$-time $(1+\varepsilon)$-approximate bichromatic 2-center algorithm.

\subsection{An IB2C Algorithm}
In this section, we solve the IB2C problem in $O(m+U^3 \log^2 U)$ time, and in turn establish our approximate bichromatic 2-center algorithm of $O(n+ (1/\varepsilon)^3 \log^2 (1/\varepsilon))$ time.
The IB2C problem itself is of independent interest, as it is a natural variant of the standard bichromatic 2-center problem.

Let $T$ be a set of $m$ point-pairs such that $P(T) \subseteq [U] \times [U]$, which is the input of the IB2C problem.
For a point $a \in [U] \times [U]$, we define $T_a \subseteq [U] \times [U]$ to be the set consisting of all points $b \in [U] \times [U]$ such that $(a,b) \in T$.
Note that $\sum_{a \in [U] \times [U]} |T_a| = O(m)$.
To solve the IB2C problem, we first compute a subset $T' \subseteq T$ with the property: a pair $(D_1,D_2)$ of disks bichromatically covers $T$ iff it bichromatically covers $T'$.
To this end, we observe an important fact.
\begin{lemma} \label{lem-cover}
    Let $(a,b_1),\dots,(a,b_k)$ be $k$ pairs of points in $\mathbb{R}^2$ sharing a common point $a$, and $b \in \mathbb{R}^2$ be a point in the convex hull of $b_1,\dots,b_k$.
    If $(a,b_1),\dots,(a,b_k)$ are all bichromatically covered by a pair $(D_1,D_2)$ of disks and $b \in D_1 \cup D_2$, then $(a,b)$ is also bichromatically covered by $(D_1,D_2)$.
\end{lemma}
\begin{proof}
Let $(D_1,D_2)$ be a pair of disks bichromatically covering $(a,b_1),\dots,(a,b_k)$ such that $b \in D_1 \cup D_2$.
We want to show that $(D_1,D_2)$ bichromatically covers $(a,b)$.
Without loss of generality, we assume $b \in D_1$.
If $a \in D_2$, we are done.
Otherwise, $a \in D_1$ and $b_1,\dots,b_k \in D_2$, because $(D_1,D_2)$ bichromatically covers $(a,b_1),\dots,(a,b_k)$.
Since $b$ lies in the convex hull of $b_1,\dots,b_k$ and $D_2$ is convex, we have $b \in D_2$.
Thus, $(D_1,D_2)$ bichromatically covers $(a,b)$.
\end{proof}

We construct $T'$ as follows.
For a set $Z \subseteq [U] \times [U]$ and a point $z \in Z$, we say $z$ is a \textit{left} (resp., \textit{right}) \textit{extreme point} in $Z$ if all points in $Z$ on the same horizontal line as $z$ are to the right (resp., left) of $z$, except $z$ itself.
Note that \textbf{(1)} $Z$ is contained in the convex hull of the left and right extreme points in $Z$ and \textbf{(2)} the number of the left/right extreme points in $Z$ is $O(U)$.
For $a \in [U] \times [U]$, let $T_a' \subseteq T_a$ be the subset consisting of all left and right extreme points in $T_a$.
Then we define $T' = \{(a,b): a \in [U] \times [U], b \in T_a'\}$.
We have $|T'| = O(U^3)$, as $T_a' = O(U)$ for all $a \in [U] \times [U]$.
Furthermore, $T'$ can be computed from $T$ in $O(m+U^3)$ time.
The desired property of $T'$ follows from the above lemma.
\begin{corollary}\label{cor:10}
    A pair $(D_1,D_2)$ of disks bichromatically covers $T$ iff it bichromatically covers $T'$.
\end{corollary}
\begin{proof}
The ``only if'' part follows directly from the fact that $T' \subseteq T$.
To see ``if'', let $(D_1,D_2)$ be a pair of disks bichromatically covering $T'$.
We want to show that $(D_1,D_2)$ bichromatically covers every element in $T$.
Consider a point-pair $(a,b) \in T$.
Since $(b,a) = (a,b) \in T$, we have $T_b \neq \emptyset$ and thus $T_b' \neq \emptyset$.
It follows that $b \in P(T')$ and $b \in D_1 \cup D_2$.
In addition, we have $b \in T_a \subseteq \mathcal{CH}(T_a')$.
Because $(D_1,D_2)$ bichromatically covers $T'$, it bichromatically covers all the pairs in $\{(a,c): c \in T_a'\}$.
By Lemma~\ref{lem-cover}, $(a,b)$ is also bichromatically covered by $(D_1,D_2)$, as $b \in \mathcal{CH}(T_a')$ and $b \in D_1 \cup D_2$.
\end{proof}

Now it suffices to solve the IB2C problem for $T'$, whose size is $O(U^3)$.
To this end, we consider the configuration of an optimal solution.
Let $r^*$ be the radius of the larger disk in an optimal solution.
\begin{lemma}\label{lem:100}
    For all $r \geq r^*$, there exists a pair $(D_1,D_2)$ of congruent disks of radius $r$ bichromatically covering $T'$ such that the centers of $D_1$ and $D_2$ are both in $[U] \times [U]$.
    In particular, $r^* \in \{\sqrt{0},\sqrt{1},\dots,\sqrt{2(U-1)^2}\}$.
\end{lemma}
\begin{proof}
Let $r \geq r^*$ be a real number.
Clearly, there exists a pair $(D_1,D_2)$ of congruent integral disks of radius $r$ bichromatically covering $T'$, since $r \geq r^*$.
Let $c_i \in [U] \times [U]$ be the point closest to the center of $D_i$, for $i \in \{1,2\}$.
One can easily verify that $([U] \times [U]) \cap D_i \subseteq ([U] \times [U]) \cap D_r(c_i)$ and hence $P(T') \cap D_i \subseteq P(T') \cap D_r(c_i)$, for $i \in \{1,2\}$.
As such, $(D_r(c_1),D_r(c_2))$ is a pair of congruent disks of radius $r$ bichromatically covering $T'$ whose centers are in $[U] \times [U]$.
This proves the first statement of the lemma.
To show the second statement, let $(D_1,D_2)$ be a pair of congruent disks of radius $r^*$ bichromatically covering $T'$ such that the centers of $D_1$ and $D_2$ are both in $[U] \times [U]$.
Since $(D_1,D_2)$ is an optimal solution, either $D_1$ or $D_2$ is unshrinkable in the sense that there is a point $a \in P(T')$ lying on its boundary.
Therefore, $r^*$ is equal to the distance between two points in $[U] \times [U]$, which implies $r^* \in \{\sqrt{0},\sqrt{1},\dots,\sqrt{2(U-1)^2}\}$.
\end{proof}

By the above lemma, we can do binary search for $r^*$ among the $O(U^2)$ values $\sqrt{0},\sqrt{1},\dots,\sqrt{2(U-1)^2}$, and pass to the decision problem, namely, deciding whether there is a feasible solution of radius $r$ for a given number $r$.
In addition, according to the above lemma, when solving the decision problem, we may require the centers of the two disks to be in $[U] \times [U]$.
Therefore, it suffices to solve the following decision problem.
\medskip

\noindent
\textbf{The Decision Problem.}
Given a set $T'$ of $O(U^3)$ pairs of points in $[U] \times [U]$ and a value $r$, decide whether there exist two points $c_1,c_2 \in [U] \times [U]$ such that $(D_r(c_1),D_r(c_2))$ bichromatically covers $T'$.
\medskip

To solve this problem, we first establish a sufficient and necessary condition for $(D_r(c_1),D_r(c_2))$ to bichromatically cover $T'$.
\begin{lemma}\label{lem:110}
    For $c_1,c_2 \in [U] \times [U]$, $(D_r(c_1),D_r(c_2))$ bichromatically covers $T'$ iff $c_1,c_2 \in \mathcal{U}_r(T')$ and $P(T') \subseteq D_r(c_1) \cup D_r(c_2)$.
\end{lemma}
\begin{proof}
Assume that $(D_r(c_1),D_r(c_2))$ bichromatically covers $T'$.
Then $c_1,c_2 \in D_r(a) \cup D_r(b)$ for every $(a,b) \in T'$, which implies $c_1,c_2 \in \mathcal{U}_r(T')$.
Also, $a,b \in D_r(c_1) \cup D_r(c_2)$ for every $(a,b) \in T'$, which implies $P(T') \subseteq D_r(c_1) \cup D_r(c_2)$.
Next, assume $c_1,c_2 \in \mathcal{U}_r(T')$ and $P(T') \subseteq D_r(c_1) \cup D_r(c_2)$.
We want to show that $(D_r(c_1),D_r(c_2))$ bichromatically covers $T'$.
It suffices to show that every pair in $T'$ is bichromatically covered by $(D_r(c_1),D_r(c_2))$.
Let $(a,b) \in T'$ be a pair.
We have $a,b \in D_r(c_1) \cup D_r(c_2)$, since $P(T') \subseteq D_r(c_1) \cup D_r(c_2)$.
Without loss of generality, assume $a \in D_r(c_1)$.
If $b \in D_r(c_2)$, we are done.
Otherwise, we must have $b \in D_r(c_1)$ and $a \in D_r(c_2)$, where the former follows from the fact $b \in D_r(c_1) \cup D_r(c_2)$ and the latter follows from the fact $c_2 \in D_r(a) \cup D_r(b)$.
Therefore, $(a,b)$ is bichromatically covered by $(D_r(c_1),D_r(c_2))$.
\end{proof}

Using Lemma~\ref{lem:110}, we solve the decision problem in two steps.
In the first step, we compute the set $C$ of all points in $[U] \times [U]$ that lie in $\mathcal{U}_r(T')$.
We call the points in $C$ \textit{candidate centers}.
In the second step, we check if there exist two candidate centers $c_1,c_2 \in C$ such that $P(T') \subseteq D_r(c_1) \cup D_r(c_2)$.
By Lemma~\ref{lem:110}, the answer of the decision problem is ``yes'' iff such two points exist.

The difficulty of the first step is that we are not able to compute $\mathcal{U}_r(T')$ efficiently, unless the points in $P(T')$ lie in a disk of radius $r$.
To resolve this issue, we recall the definition of $T_a'$ for a point $a \in [U] \times [U]$.
We observe that
\begin{equation*}
    \mathcal{U}_r(T') = \bigcap_{(a,b) \in T'} D_r(a) \cup D_r(b) = \bigcap_{a \in P(T')} D_r(a) \cup \mathcal{I}_r(T_a').
\end{equation*}
Therefore, a point is in $\mathcal{U}_r(T')$ iff it is in $D_r(a) \cup \mathcal{I}_r(T_a')$ for all $a \in P(T')$.
Note that we can compute $\mathcal{I}_r(T_a')$ for all $a \in P(T')$ in $O(U^3 \log U)$ time because $\sum_{a \in P(T')} |T_a'| = O(U^3)$.
With $\mathcal{I}_r(T_a')$, a direct way to compute the candidate centers is to check for every $c \in [U] \times [U]$ whether $c \in D_r(a) \cup \mathcal{I}_r(T_a')$ for all $a \in P(T')$.
However, this requires $\Omega(U^4)$ time since $|P(T')| = \Omega(U^2)$ in the worst case.
In order to do it more efficiently, our idea is to compute the candidate centers in a row simultaneously.
Formally, let $R_j = [U] \times \{j\}$ be the set of the points in the $j$-th row of $[U] \times [U]$, for $j \in [U]$.
We want to find the candidate centers in $R_j$.
For $a \in P(T')$, let $I_a$ be the intersection of $D_r(a) \cup \mathcal{I}_r(T_a')$ and the horizontal line $y = j$.
We define the \textit{depth} of a point in $R_j$ as the number of $I_a$'s containing it.
A point in $R_j$ is a candidate center iff it is contained in $I_a$ for all $a \in P(T')$, or equivalently, its depth is $|P(T')|$.
We find the candidate centers in $R_j$ by computing the depths of these points as follows.
Since $D_r(a)$ and $\mathcal{I}_r(T_a')$ are both convex, each $I_a$ is either an interval or a double-interval (i.e., the union of two disjoint intervals).
Let $E$ be the set of the endpoints of these intervals and double-intervals.
Note that $|E| \leq 4|P(T')| = O(U^2)$, as an interval has two endpoints and a double-interval has four.
We sort the points in $E \cup R_j$ from left to right in $O(U^2 \log U)$ time, and scan these points in this order.
In this procedure, we maintain a number $\mathit{dep}$ which is the depth of the current point.
Initially, we set $\mathit{dep} = 0$.
At every time we hit a left (resp., right) endpoint in $E$, we increase (resp., decrease) $\mathit{dep}$ by 1.
When we hit a point in $R_j$, its depth is just the current value of $\mathit{dep}$.
In this way, we compute the depths of the points in $R_j$ in $O(U^2 \log U)$, and find the candidate centers in $R_j$.
After doing this for all $j \in [U]$, we obtain the set $C$ of all candidate centers, which takes $O(U^3 \log U)$ time in total.

With $C$, we proceed to the second step, namely, finding two candidate centers $c_1,c_2 \in C$ such that $P(T') \subseteq D_r(c_1) \cup D_r(c_2)$, or deciding the nonexistence of them.
To this end, we first build a set of data structures $\mathcal{E}_1,\dots,\mathcal{E}_U$ as follows.
For each row $R_j$, consider the points in $C_j = C \cap R_j$.
These points lie on the horizontal line $\ell_j: y = j$.
The data structure $\mathcal{E}_j$ can answer 1-dimensional range-emptiness queries on $C_j$: given an interval $I$ on the line $\ell_j$, $\mathcal{E}_j$ can decide whether $C_j \cap I = \emptyset$ and return a point in $C_j \cap I$ if $C_j \cap I \neq \emptyset$.
Such a data structure is well-known, and can be built in $O(U \log U)$ time with $O(\log U)$ query time, as $|C_j| = O(U)$.
Building all $\mathcal{E}_1,\dots,\mathcal{E}_U$ takes $O(U^2 \log U)$ time.
With the data structures in hand, we solve the problem by considering the rows $R_1,\dots,R_U$ separately.
For each row $R_j$, we want to check whether there exist $c_1 \in C \cap R_j$ and $c_2 \in C$ such that $P(T') \subseteq D_r(c_1) \cup D_r(c_2)$.
Fix $j \in [U]$.
For $i \in [U]$, define $p_i \in R_j$ as the point whose coordinate is $(i,j)$.
Assume we set $c_1 = p_i$ for some $p_i \in C \cap R_j$.
Then there exists $c_2 \in C$ satisfying the desired property iff $C \cap \mathcal{I}_r(P(p_i)) \neq \emptyset$ where $P(p_i) = P(T') \backslash D_r(p_i)$.
Note that $C \cap \mathcal{I}_r(P(p_i)) = \bigcup_{j' \in [U]} (C_{j'} \cap I_{i,j'})$ where $I_{i,j'} = \mathcal{I}_r(P(p_i)) \cap \ell_{j'}$ is an interval on the line $\ell_{j'}$.
Therefore, if we know $I_{i,j'}$ for all $j' \in [U]$, then we can use the data structures $\mathcal{E}_1,\dots,\mathcal{E}_U$ to determine in $O(U \log U)$ time the emptiness of $C_{j'} \cap I_{i,j'}$ for all $j' \in [U]$ and hence the emptiness of $C \cap \mathcal{I}_r(P(p_i))$; furthermore, if $C \cap \mathcal{I}_r(P(p_i)) \neq \emptyset$, a point $c_2 \in C \cap \mathcal{I}_r(P(p_i))$ can be found by one of $\mathcal{E}_1,\dots,\mathcal{E}_U$.
It follows that as long as we know $I_{i,j'}$ for all $i,j' \in [U]$, we can determine in $O(U^2 \log U)$ time if there exist $c_1 \in C \cap R_j$ and $c_2 \in C$ such that $P(T') \subseteq D_r(c_1) \cup D_r(c_2)$, by enumerating all $p_i \in C \cap R_j$.

It now suffices to show how to compute the intervals $I_{i,j'}$ in $O(U^2\log U)$ time.
To this end, we first introduce some notations.
For each $p_i \in R_j$, define $P'(p_i) = \{(x,y) \in P(p_i): x \leq i\}$ and $P''(p_i) = \{(x,y) \in P(p_i): x \geq i\}$.
Then we have $P(p_i) = P'(p_i) \cup P''(p_i)$.
\begin{lemma} \label{lem-3properties}
    The sets $P'(p_i)$ and $P''(p_i)$ have the following properties. \\
    \textnormal{\bf (1)} $P(p_i) = P'(p_i) \cup P''(p_i)$ for all $i \in [U]$. \\
    \textnormal{\bf (2)} $P'(p_1) \subseteq \cdots \subseteq P'(p_U)$. \\
    \textnormal{\bf (3)} $P''(p_1) \supseteq \cdots \supseteq P''(p_U)$.
\end{lemma}
\begin{proof}
The property \textbf{(1)} follows directly from definition.
The properties \textbf{(2)} and \textbf{(3)} are symmetric, so it suffices to show \textbf{(2)}.
Let $a \in P'(p_i)$ be a point.
We want to show $a \in P'(p_{i+1})$.
Since $a \in P'(p_i)$, we know that \textbf{(i)} $a \in P(T')$, \textbf{(ii)} $a \notin D_r(p_i)$, and \textbf{(iii)} $a$ is to the left of $p_i$ or has the same $x$-coordinate as $p_i$.
The fact \textbf{(iii)} implies that $p_{i+1}$ is farther away from $a$ than $p_i$.
Combining this fact with the fact \textbf{(ii)}, we have $a \notin D_r(p_{i+1})$.
Also the fact \textbf{(iii)} implies that $a$ is to the left of $p_{i+1}$.
Therefore, we have $a \in P'(p_{i+1})$, completing the proof.
\end{proof}

\noindent
By the property \textbf{(1)} in the above lemma, we have for all $i,j' \in [U]$,
\begin{equation*}
    I_{i,j'} = (\mathcal{I}_r(P'(p_i)) \cap \mathcal{I}_r(P''(p_i))) \cap \ell_{j'} = I_{i,j'}' \cap I_{i,j'}'',
\end{equation*}
where $I_{i,j'}' = \mathcal{I}_r(P'(p_i)) \cap \ell_{j'}$ and $I_{i,j'}'' = \mathcal{I}_r(P''(p_i)) \cap \ell_{j'}$.
Therefore, it suffices to compute $I_{i,j'}'$ and $I_{i,j'}''$ for all $i,j' \in [U]$.
In fact, we only need to show how to compute $I_{i,j'}'$, since $I_{i,j'}''$ can be computed in the same way.
We shall use the property \textbf{(2)} in Lemma~\ref{lem-3properties}.
Define $Q_1 = P'(p_1)$ and $Q_i = P'(p_i) \backslash P'(p_{i-1})$ for $i \in \{2,\dots,U\}$.
By the property \textbf{(2)}, we have $P'(p_i) = P'(p_{i-1}) \cup Q_i$ and hence $\mathcal{I}_r(P'(p_i)) = \mathcal{I}_r(P'(p_{i-1})) \cap \mathcal{I}_r(Q_i)$ for $i \in \{2,\dots,U\}$.
It follows that $I_{i,j'}' = I_{i-1,j'}' \cap \mathcal{I}_r(Q_i)$ for all $i \in \{2,\dots,U\}$ and $j' \in [U]$.

Based on this fact, we can compute the intervals $I_{i,j'}'$ as follows.
First, we compute $Q_1,\dots,Q_U$.
To this end, we determine for each point $a \in P(T')$ the leftmost point $\xi(a) \in R_j = \{p_1,\dots,p_U\}$ such that $a \notin D_r(\xi(a))$ and $\xi(a)$ is not to the left of $a$.
Note that $\xi(a)$ can be computed in $O(\log U)$ time by applying binary search among the points in $R_j$.
Also note that if $\xi(a) = p_i$, then $a \in Q_i$.
Therefore, once we know $\xi(a)$ for all $a \in P(T')$, we obtain $Q_1,\dots,Q_U$.
This step takes $O(U^2 \log U)$ times since $|P(T')| = O(U^2)$.
Next, we compute $\mathcal{I}_r(Q_1),\dots,\mathcal{I}_r(Q_U)$.
This can be done in $O(U^2 \log U)$ time, since $\sum_{i=1}^U |Q_i| \leq U^2$ (note that $Q_1,\dots,Q_U$ are disjoint).
Since each $\mathcal{I}_r(Q_i)$ is convex and has complexity $O(|Q_i|)$, we can do horizontal line intersection (HLI) on $\mathcal{I}_r(Q_i)$ in $O(\log |Q_i|)$ time (which is $O(\log U)$ time since $|Q_i|=O(U^2)$), that is, given a horiozontal line $\ell$, we can compute the intersection $\mathcal{I}_r(Q_i) \cap \ell$ in $O(\log |Q_i|)$ time~\cite{ref:HershbergerFi91}, by using binary search to find the two intersection points of the boundary of $\mathcal{I}_r(Q_i)$ and $\ell$.
With $\mathcal{I}_r(Q_1),\dots,\mathcal{I}_r(Q_U)$ in hand, we are ready to compute the intervals $I_{i,j'}'$.
For each $j' \in [U]$, we compute $I_{1,j'}',\dots,I_{U,j'}'$ in order.
We have $I_{1,j'}' = \mathcal{I}_r(Q_1) \cap \ell_j$ and $I_{i,j'}' = I_{i-1,j'}' \cap \mathcal{I}_r(Q_i)$ for all $i \in \{2,\dots,U\}$.
It turns out that each $I_{i,j'}'$ can be computed in $O(\log U)$ time.
Indeed, $I_{1,j'}'$ is computed via an HLI mentioned above, and each $I_{i,j'}'$ for $i \in \{2,\dots,U\}$ is computed by first computing $\mathcal{I}_r(Q_i) \cap \ell_{j'}$ (which is a segment) and then intersecting it with $I_{i-1,j'}'$.
Therefore, for a fixed $j' \in [U]$, we can computing $I_{1,j'}',\dots,I_{U,j'}'$ in $O(U\log U)$ time, which implies that $I_{i,j'}'$ for all $i,j' \in [U]$ can be computed in $O(U^2\log U)$ time as long as we know $\mathcal{I}_r(Q_1),\dots,\mathcal{I}_r(Q_U)$.
Including the time for computing $\mathcal{I}_r(Q_1),\dots,\mathcal{I}_r(Q_U)$, we see that computing $I_{i,j'}'$ for all $i,j' \in [U]$ takes $O(U^2 \log U)$ time, so does computing $I_{i,j'}$ for all $i,j' \in [U]$.

By our previous argument, it follows that one can decide whether there exists $c_1 \in C \cap R_j$ and $c_2 \in C$ satisfying $P(T') \subseteq D_r(c_1) \cup D_r(c_2)$ in $O(U^2 \log U)$ time.
By considering all $j \in [U]$, we can complete the second step in $O(U^3 \log U)$ time.

Now we see that both steps can be done in $O(U^3 \log U)$ time, which is also the time for solving the decision version of the IB2C problem.
Using the decision algorithm as a sub-routine to do binary search, we can solve the IB2C problem on $T'$ in $O(U^3 \log^2 U)$ time.
Including the time for constructing $T'$ from $T$, we finally obtain an IB2C algorithm with $O(m+U^3 \log^2 U)$ running time.
\begin{theorem}
    There exists an $O(m+U^3 \log^2 U)$-time IB2C algorithm.
\end{theorem}
\begin{corollary}
    The $(1+\varepsilon)$-approximate bichromatic 2-center problem on a set of $n$ pairs of points in the plane is solvable in $O(n+ (1/\varepsilon)^3 \log^2 (1/\varepsilon))$ time.
\end{corollary}




%

\bibliographystyle{plain}
\bibliography{reference}




\end{document}